\documentclass[a4paper]{article}
\usepackage[margin=30mm]{geometry}

\usepackage[T1]{fontenc}
\usepackage{graphicx}
\usepackage[algoruled,linesnumbered,algo2e,vlined]{algorithm2e}
\usepackage{amsmath,amssymb}
\usepackage{url}
\usepackage{rotating}
\usepackage{comment}
\usepackage{multirow}
\usepackage{amsthm}
\theoremstyle{plain}
\newtheorem{theorem}{Theorem}

\newtheorem{lemma}[theorem]{Lemma}

\newtheorem{proposition}[theorem]{Proposition}

\newcommand{\emptystr}{\varepsilon}

\newcommand{\lcp}{\mathit{lcp}} % longest common prefix

\newcommand{\gtext}{\mathsf{T}}
\newcommand{\sSigma}{\Sigma_{s}}
\newcommand{\pSigma}{\Sigma_{p}}
\newcommand{\ssigma}{\sigma_{s}}
\newcommand{\psigma}{\sigma_{p}}
\newcommand{\penc}[1]{\langle#1\rangle}

\newcommand{\zeros}[1]{|#1|_{0}}
\newcommand{\multiset}[1]{\{\!\!\{#1\}\!\!\}}
\newcommand{\GroupA}{\mathit{G}}
\newcommand{\ZeroA}{\mathit{Z}}
\newcommand{\EndA}{\mathit{E}}

\newcommand{\psa}{\tau}

\newcommand{\LF}{\mathit{LF}}

\newcommand{\Lstr}{\mathsf{L}}

\newcommand{\Pref}{\mathit{Pref}}
\newcommand{\interval}[1]{\mathit{I}(#1)}

\newcommand{\idtt}[1]{\ensuremath{\mathtt{#1}}}

\usepackage{hyperref}
\usepackage{authblk}

\title{Inverting Parameterized Burrows-Wheeler Transform}

\author[1]{Shogen Kawanami\thanks{kawanami.syogen518@mail.kyutech.jp}}
\author[1]{Kento Iseri\thanks{iseri.kento210@mail.kyutech.jp}}
\author[1]{Tomohiro I\thanks{tomohiro@ai.kyutech.ac.jp}}
\affil[1]{Kyushu Institute of Technology, 680-4 Kawazu, Iizuka, Fukuoka 820-8502, Japan}

\date{}

\begin{document}
\maketitle

\begin{abstract}
The Burrows-Wheeler Transform (BWT) of a string is an invertible permutation of the string,
which can be used for data compression and compact indexes for string pattern matching.
Ganguly et al. [SODA, 2017] introduced the parameterized BWT (pBWT) to design
compact indexes for parameterized matching (p-matching), 
a variant of string pattern matching with parameter symbols introduced by Baker [STOC, 1993].
Although the pBWT was inspired by the BWT, it is not obvious whether the pBWT itself is invertible or not.
In this paper we show that we can retrieve the original string (up to renaming of parameter symbols)
from the pBWT of length $n$ in $O(n^2)$ time and $O(n)$ space.
\end{abstract}

\section{Introduction}\label{sec:intro}
A \emph{parameterized string} (\emph{p-string}) consists of two different types of symbols, static symbols (s-symbols) and parameter symbols (p-symbols).
Let $\sSigma$ and $\pSigma$ denote the sets of s-symbols and p-symbols, respectively.
Two p-strings $x$ and $y$ are said to \emph{parameterized match} (\emph{p-match}) if and only if
$x$ can be transformed into $y$ by applying a bijection on $\pSigma$ to every occurrence of p-symbols in $x$.
For example with $\sSigma = \{ \idtt{A}, \idtt{B} \}$ and $\pSigma = \{ \idtt{x}, \idtt{y}, \idtt{z} \}$, 
two p-strings $\idtt{xAyxzzyB}$ and $\idtt{zAxzyyxB}$ p-match because
$\idtt{xAyxzzyB}$ can be transformed into $\idtt{zAxzyyxB}$ 
by replacing $\idtt{x}, \idtt{y}$ and $\idtt{z}$ with $\idtt{z}$, $\idtt{x}$ and $\idtt{y}$, respectively.
P-matching was introduced by Baker to address a variant of string pattern matching with renamable symbols in software maintenance and plagiarism detection~\cite{1993Baker_TheorOfParamPatterMatch_STOC,1996Baker_ParamPatterMatchAlgorAnd,1997Baker_ParamDuplicInStrinAlgor}.
Later it finds wider applications such as detecting similar images and isomorphic graphs, 
and algorithms for p-matching have been extensively studied 
(see a recent survey~\cite{2020MendivelsoTP_BriefHistorOfParamMatch} and references therein).

Almost all efficient p-matching algorithms are based on the \emph{prev-encoding}~\cite{1993Baker_TheorOfParamPatterMatch_STOC} 
that replaces every occurrence of a p-symbol in a p-string 
with the distance to the previous occurrence of the same p-symbol, or with ``0'' if there is no previous occurrence.
By definition, two p-strings p-match if and only if their prev-encoded strings are identical.
For example, the p-matching p-strings $\idtt{xAyxzzyB}$ and $\idtt{zAxzyyxB}$ are both prev-encoded to $\idtt{0A03014B}$.
Also, the prefixes of the prev-encoding do not change when symbols are appended.
These properties enable us to extend suffix trees and suffix arrays for indexes of p-matching:
For a p-string $\gtext$ of length $n$, the parameterized suffix tree (pST)~\cite{1993Baker_TheorOfParamPatterMatch_STOC} is 
the compacted trie of the prev-encoded suffixes of $\gtext$,
and the parameterized suffix array (pSA)~\cite{2008DeguchiHBIT_ParamSuffixArrayForBinar_PSC}, denoted by $\psa[1..n]$ in this paper, 
stores the starting positions of prev-encoded suffixes in their lexicographic order.
On the other hand, the suffixes of the prev-encoding may change when symbols are prepended,
which makes some algorithms for p-matching more complicated than those for exact string pattern matching.

The Burrows-Wheeler Transform (BWT)~\cite{Burrows1994BWT} of a string is an invertible permutation of the string.
The BWT was originally introduced for data compression but has become the basis of compact or compressed indexes for string pattern matching
such as the FM-index~\cite{Ferragina2000ODS} and the r-index~\cite{2018GagieNP_OptimTimeTextIndexIn_SODA}.
A notable property of the BWT is that it enables us to implement the \emph{LF-mapping}
that maps the lexicographic rank of a suffix of the original string to the rank of the suffix longer by one.

Inspired by the FM-index, Ganguly et al.~\cite{2017GangulyST_PbwtAchievSuccinDataStruc_SODA} introduced the parameterized BWT (pBWT) 
to design a compact index for p-matching.
For a p-string $\gtext$ of length $n$ over $(\sSigma \cup \pSigma)$, 
the pBWT $\Lstr$ of $\gtext$ is a string of length $n$ over an alphabet of size $\sigma = |\sSigma \cup \pSigma|$, 
and thus, within compact space of $O(n \log \sigma)$ bits.
For any position $i$ with $\psa[i] \neq 1$,
$\Lstr[i]$ holds the necessary and sufficient information 
to turn the prev-encoding of $\gtext[\psa[i]..n]$ into that of $\gtext[\psa[i]-1..n]$ (see Section~\ref{sec:prelim} for its formal definition).
So it is not difficult to see that the prev-encoding of $\gtext$ is retrieved from $\Lstr$ if the LF-mapping is provided.
However it is not obvious whether we can retrieve the information of the LF-mapping from $\Lstr$ alone
because existing pBWT-based indexes~\cite{2017GangulyST_PbwtAchievSuccinDataStruc_SODA,2021KimC_SimplFmIndexForParam,2024IseriIHKYS_BreakBarrierInConstCompac_ICALP} 
use other data structures together with $\Lstr$ to implement the LF-mapping.

In this paper, we show that we can retrieve the information of the LF-mapping and $\gtext$ (up to renaming of p-symbols)
from $\Lstr$ in $O(n^2)$ time and $O(n)$ space.
To the best of our knowledge, the invertibility of pBWTs has not been explored before
and our result paves the way to using pBWTs as a compression method for p-strings.

\subsection{Related work}
The problem we tackled in this work can also be seen as a reverse engineering problem for string data structures,
which have been actively studied to gain a deeper understanding of data structures (e.g.,~\cite{2003BannaiIST_InferStrinFromGraphAnd_MFCS,2013KucherovTV_CombinOfSuffixArray,2014IIBT_InferStrinFromSuffixTrees,2014CazauxR_ReverEnginOfCompacSuffix,Starikovskaya2015sto,2023KaerkkaeinenPP_StrinInferFromLongesCommon}).
For those related to p-matching, there are studies on parameterized border arrays~\cite{2011IIBT_VerifAndEnumerParamBorder}
and parameterized suffix and LCP arrays~\cite{2024AmirKLMS_ReconParamStrinFromParam,2024AmirKMS_LinearTimeReconOfParam_SPIRE}.
Since pBWTs sometimes act differently from BWTs,
a deeper understanding of its properties would be a key to further development in pBWT-based data structures like~\cite{2022GangulyST_FullyFunctParamSuffixTrees_ICALP}.

\section{Preliminaries}\label{sec:prelim}

\subsection{Notations and tools}
An integer interval $\{ i, i+1, \dots, j\}$ is denoted by $[i..j]$, 
where $[i..j]$ represents the empty interval if $i > j$.
Also $[i..j)$ denotes $[i..j-1]$.

Let $\Sigma$ be an ordered finite \emph{alphabet}.
An element of $\Sigma^*$ is called a \emph{string} over $\Sigma$.
The length of a string $w$ is denoted by $|w|$. 
The empty string $\emptystr$ is the string of length 0,
that is, $|\emptystr| = 0$.
Let $\Sigma^+ = \Sigma^* - \{\emptystr\}$ and $\Sigma^k = \{ w \in \Sigma^* \mid |w| = k \}$ for any non-negative integer $k$.
The concatenation of two strings $x$ and $y$ is denoted by $x \cdot y$ or simply $xy$.
When a string $w$ is represented by the concatenation of strings $x$, $y$ and $z$ (i.e., $w = xyz$), 
then $x$, $y$ and $z$ are called a \emph{prefix}, \emph{substring}, and \emph{suffix} of $w$, respectively.
A substring $x$ of $w$ is called \emph{proper} if $x \neq w$.

The $i$-th symbol of a string $w$ is denoted by $w[i]$ for $1 \leq i \leq |w|$,
and the substring of a string $w$ that begins at position $i$ and
ends at position $j$ is denoted by $w[i..j]$ for $1 \leq i \leq j \leq |w|$,
i.e., $w[i..j] = w[i]w[i+1] \cdots w[j]$.
For convenience, let $w[i..j] = \emptystr$ if $j < i$.
For two strings $x$ and $y$, let $\lcp(x, y)$ denote the length of the longest common prefix between $x$ and $y$.
We consider the lexicographic order over $\Sigma^*$ by extending the strict total order $<$ defined on $\Sigma$:
$x$ is lexicographically smaller than $y$ (denoted as $x < y$) if and only if 
either $x$ is a proper prefix of $y$ or $x[\lcp(x, y)+1] < y[\lcp(x, y)+1]$ holds.

We will use the following result for dynamic strings as a tool:
\begin{lemma}[\cite{2015MunroN_ComprDataStrucForDynam_ESA}]\label{lem:drs}
  For a positive integer $U$, a string over an alphabet $[1..U]$ can be dynamically maintained
  while supporting insertion/deletion of a symbol to/from any position of the string as well as random access 
  in $(m + o(m)) \lg U$ bits of space and $O(\frac{\lg m}{\lg \lg m})$ query and update times, 
  where $m$ is the length of the current string.
\end{lemma}

\subsection{Parameterized strings}
Let $\sSigma$ and $\pSigma$ denote two disjoint sets of symbols.
We call a symbol in $\sSigma$ a \emph{static symbol} (\emph{s-symbol}) 
and a symbol in $\pSigma$ a \emph{parameter symbol} (\emph{p-symbol}).
For a technical reason, $\sSigma$ is considered to be disjoint with the set of integers
and we assume that an s-symbol is lexicographically smaller than any integer.
Let $\ssigma = |\sSigma|$, $\psigma = |\pSigma|$ and $\sigma = \ssigma + \psigma$.
A \emph{parameterized string} (\emph{p-string}) is a string over $(\sSigma \cup \pSigma)$.

For any p-string $w$ of length $n$, the \emph{prev-encoding} $\penc{w}$ of $w$ is the string in $(\sSigma \cup [0..n))^{n}$ 
such that, for any $i~(1 \le i \le n)$,
\begin{equation*}
  \penc{w}[i] =
  \begin{cases}
    w[i]     & \mbox{if $w[i] \in \sSigma$}, \\
    0        & \mbox{if $w[i] \in \pSigma$ and $w[i]$ does not appear in $w[1..i-1]$}, \\
    i - j    & \mbox{otherwise,}
  \end{cases}
\end{equation*}
where $j$ is the largest position in $[1..i-1]$ with $w[i] = w[j]$.
By prev-encoding, all the occurrences of a p-symbol are connected by a single chain
on the array storing the distance to the previous occurrence of the same p-symbol, which
ends with the symbol $0$ at the leftmost occurrence of the p-symbol.
Let $\zeros{\penc{w}}$ denote the the number of $0$ in $\penc{w}$, which is equal to the number of distinct p-symbols in $w$.
For example, $\penc{w} = \idtt{0A03014B}$ and $\zeros{\penc{w}} = 3$ 
for $w = \idtt{xAyxzzyB}$ with $\sSigma = \{ \idtt{A}, \idtt{B} \}$ and $\pSigma = \{ \idtt{x}, \idtt{y}, \idtt{z} \}$.

The following properties hold for prev-encoding $\penc{w}$ of a p-string $w$ of length $n$:
\begin{itemize}
  \item A p-string $w'$ p-matches $w$ if and only if $\penc{w'} = \penc{w}$. \label{PPmatch}
  \item For any $1 \le i \le n$, $\penc{w}[1..i] = \penc{w[1..i]}$. \label{PPprefix}
  \item If $c \in \sSigma$, then $\penc{cw} = c\penc{w}$. \label{PPs}
  \item If $c \in \pSigma$ does not appear in $w$, then $\penc{cw} = 0\penc{w}$. \label{PPpnew}
  \item If $c \in \pSigma$ and $d$ is the leftmost occurrence of $c$ in $w$, 
    then $\penc{cw} = 0 \cdot \penc{w}[1..d-1] \cdot d \cdot \penc{w}[d+1..n]$. \label{PPp}
\end{itemize}

It is easy to see that the prev-encoding has enough information 
to retrieve the p-strings that p-match the original p-string
and we can compute the lexicographically smallest one in linear time:
\begin{lemma}\label{lemma:lex-smallest}
  Given prev-encoding $W$ of a p-string $w$ of length $n$, 
  we can compute the lexicographically smallest p-string $w'$ that p-matches $w$ in $O(n)$ time.
\end{lemma}
\begin{proof}
  For increasing positions $i$, we determine $w'[i]$ as follows:
  \begin{itemize}
    \item If $W[i] \in \sSigma$, set $w'[i] = W[i]$.
    \item If $W[i] = 0$, set $w'[i]$ to be the lexicographically smallest p-symbol that has not been used.
    \item If $W[i] = d \in [1..n)$, set $w'[i] = w'[i - d]$.
  \end{itemize}
  Since the p-strings that p-match $w$ differ up to renaming of p-symbols and 
  we choose lexicographically smaller p-symbols in the order of their leftmost occurrences,
  $w'$ becomes the lexicographically smallest.
  It is clear that it runs in linear time.
\end{proof}

Throughout this paper, we consider a p-string $\gtext$ of length $n$ over $(\sSigma \cup \pSigma)$
with the condition that $\gtext[n] = \$$, where $\$$ is the smallest s-symbol that does not appear anywhere else in $\gtext$.
For any integer $i$, let $\gtext_{i}$ denote the cyclic rotation of $\gtext$ such that $\gtext_{i} = \gtext[i'..n]\gtext[1..i'-1]$ for $i' \in [1..n]$ with $i' \equiv i \mod n$.
Let $\psa$ be the permutation of $[1..n]$ such that $\penc{\gtext_{\psa[i]}}$ is the lexicographically $i$-th string in $\{ \penc{\gtext_{i'}} \}_{i'=1}^{n}$.
Due to the existence of end-marker $\$$, $\psa$ is equivalent to the \emph{parameterized suffix array (pSA)} of $\gtext$.

The \emph{parameterized Burrows Wheeler Transform (pBWT)} $\Lstr$ of $\gtext$ is 
a string of length $n$ over the alphabet $(\sSigma \cup [1..\psigma])$
such that, for any $i~(1 \le i \le n)$,
\begin{equation*}
  \Lstr[i] =
  \begin{cases}
    \gtext_{\psa[i]}[n]                   & \mbox{if $\gtext_{\psa[i]}[n] \in \sSigma$}, \\
    \zeros{\penc{\gtext_{\psa[i]}[1..j]}} & \mbox{otherwise,}
  \end{cases}
\end{equation*}
where $j$ is the smallest position in $\gtext_{\psa[i]}$ with $\gtext_{\psa[i]}[n] = \gtext_{\psa[i]}[j]$.

Note that $\Lstr[i]$ holds the information to turn $\penc{\gtext_{\psa[i]}}[1..\ell]$ into $\penc{\gtext_{\psa[i]-1}}[1..\ell+1]$:
\begin{proposition}\label{prop:prepend}
  For any position $i$ and $\ell \in [0..n)$, let $W = \penc{\gtext_{\psa[i]}}[1..\ell]$ and $\hat{W} = \penc{\gtext_{\psa[i]-1}}[1..\ell+1]$.
  The following statements hold (see Fig.~\ref{prop:prepend} for illustration):
  \begin{itemize}
    \item If $\Lstr[i] \in \sSigma$, then $\hat{W} = \Lstr[i]W$.
    \item If $\Lstr[i] \in [1..\psigma]$ and $\Lstr[i] > \zeros{W}$, then $\hat{W} = 0W$.
    \item If $\Lstr[i] \in [1..\psigma]$ and $d$ is the $\Lstr[i]$-th occurrence of $0$ in $W$,
      then $\hat{W} = 0 \cdot W[1..d-1] \cdot d \cdot W[d+1..\ell]$.
  \end{itemize}
\end{proposition}

\begin{figure}[t]
  \center{%
    \includegraphics[scale=0.38]{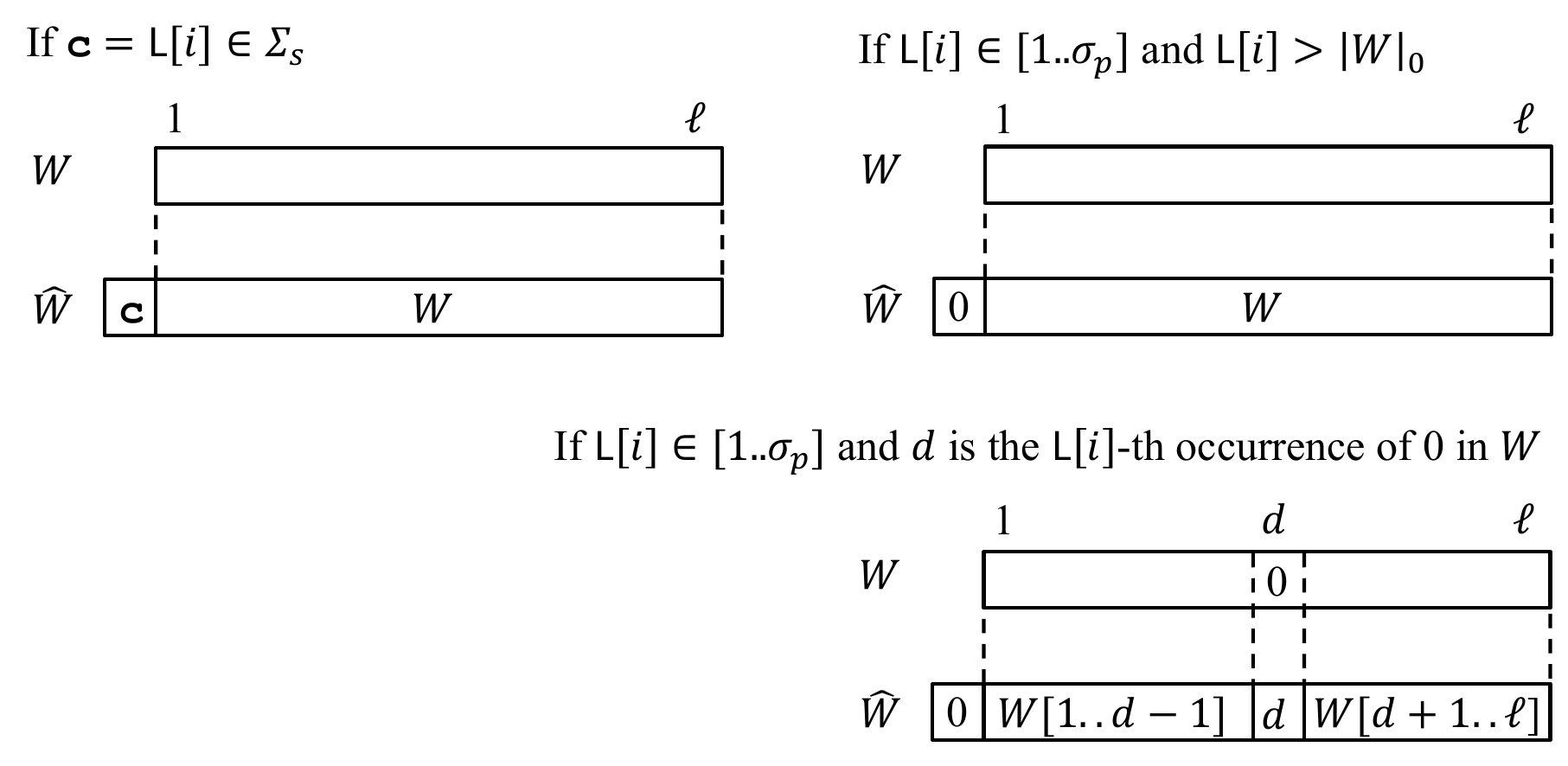}
  }
  \caption{Illustrations for the cases of Proposition~\ref{prop:prepend}.
  }
  \label{fig:prepend}
\end{figure}

The \emph{LF-mapping} $\LF$ for a p-string $\gtext$ is the permutation of $[1..n]$ such that 
$\psa[\LF[i]] \equiv \psa[i] - 1 \mod n$.

Table~\ref{table:arrays} shows an example of $\psa[i]$, $\LF[i]$ and $\Lstr[i]$.
\begin{table}[t]
  \centering
  \begin{tabular}{|c||l|l||c|c|c|l|l|}
  \hline
  \multicolumn{1}{|c||}{$i$} & \multicolumn{1}{c|}{$\gtext_i$} & \multicolumn{1}{c||}{$\penc{\gtext_i}$} & \multicolumn{1}{c|}{$\psa[i]$} & \multicolumn{1}{c|}{$\LF[i]$} & \multicolumn{1}{c|}{$\Lstr[i]$} & \multicolumn{1}{c|}{$\penc{\gtext_{\psa[i]}}$} & \multicolumn{1}{c|}{$\gtext_{\psa[i]}$} \\ \hline
  1   & $\idtt{xyxzzxxyx\$}$  & $\idtt{002013162\$}$  & 10  & 2   & 1            & $\idtt{\$002013162}$  & $\idtt{\$xyxzzxxyx}$ \\
  2   & $\idtt{yxzzxxyx\$x}$  & $\idtt{00013162\$2}$  & 9   & 3   & 2            & $\idtt{0\$20201316}$  & $\idtt{x\$xyxzzxxy}$ \\
  3   & $\idtt{xzzxxyx\$xy}$  & $\idtt{0013102\$24}$  & 8   & 7   & 2            & $\idtt{00\$2420131}$  & $\idtt{yx\$xyxzzxx}$ \\
  4   & $\idtt{zzxxyx\$xyx}$  & $\idtt{010102\$242}$  & 2   & 8   & 2            & $\idtt{00013162\$2}$  & $\idtt{yxzzxxyx\$x}$ \\
  5   & $\idtt{zxxyx\$xyxz}$  & $\idtt{00102\$2429}$  & 5   & 9   & 1            & $\idtt{00102\$2429}$  & $\idtt{zxxyx\$xyxz}$ \\
  6   & $\idtt{xxyx\$xyxzz}$  & $\idtt{0102\$24201}$  & 3   & 4   & 3            & $\idtt{0013102\$24}$  & $\idtt{xzzxxyx\$xy}$ \\
  7   & $\idtt{xyx\$xyxzzx}$  & $\idtt{002\$242013}$  & 7   & 10  & 1            & $\idtt{002\$242013}$  & $\idtt{xyx\$xyxzzx}$ \\
  8   & $\idtt{yx\$xyxzzxx}$  & $\idtt{00\$2420131}$  & 1   & 1   & $\idtt{\$}$  & $\idtt{002013162\$}$  & $\idtt{xyxzzxxyx\$}$ \\  
  9   & $\idtt{x\$xyxzzxxy}$  & $\idtt{0\$20201316}$  & 4   & 6   & 2            & $\idtt{010102\$242}$  & $\idtt{zzxxyx\$xyx}$ \\
  10  & $\idtt{\$xyxzzxxyx}$  & $\idtt{\$002013162}$  & 6   & 5   & 3            & $\idtt{0102\$24201}$  & $\idtt{xxyx\$xyxzz}$ \\
  \hline
  \end{tabular}
  \caption{An example of $\psa[i]$, $\LF[i]$ and $\Lstr[i]$ for a p-string $\gtext=\idtt{xyxzzxxyx\$}$ 
  with $\sSigma = \{ \idtt{\$} \}$ and $\pSigma = \{\idtt{x}, \idtt{y}, \idtt{z} \}$.}
  \label{table:arrays}
\end{table}

\section{Algorithms}\label{sec:algorithms}
In this section, we present our algorithms for the inversion problem of pBWTs.\@

We first present a simple algorithm to see that pBWTs are invertible.
\begin{theorem}\label{theorem:n3}
  Given the pBWT $\Lstr$ of a p-string $\gtext$ of length $n$, 
  we can compute in $O(n^3)$ time and $O(n^2)$ space a p-string that p-matches $\gtext$.
\end{theorem}
\begin{proof}
  For any $\ell \ge 0$ let $\Pref_\ell$ denote the multiset of prev-encoded strings of the form $\penc{\gtext_{i}}[1..\ell]$, 
  i.e., $\Pref_{\ell} = \multiset{\penc{\gtext_{i}}[1..\ell]}_{i = 1}^{n}$.
  We show how to compute $\Pref_{\ell+1}$ from $\Pref_\ell$ in $O(n \ell) = O(n^2)$ time.
  First we sort $\Pref_\ell$ in $O(n \ell) = O(n^2)$ time by radix sort with radix of size $\Theta(n)$.
  Then the $i$-th element of the sorted $\Pref_\ell$ is equivalent to $\penc{\gtext_{\psa[i]}}[1..\ell]$.
  Since Proposition~\ref{prop:prepend} enables us to compute $\penc{\gtext_{\psa[i]-1}}[1..\ell+1]$ from $\penc{\gtext_{\psa[i]}}[1..\ell]$ in $O(\ell)$ time,
  we get $\Pref_{\ell+1}$ in $O(n \ell) = O(n^2)$ time. 
  Fig.~\ref{fig:n3} shows an example of the steps of the process.

  Repeating the above process from $\Pref_{0} = \multiset{\emptystr}_{i = 1}^{n}$ to $\Pref_{n} = \multiset{\penc{\gtext_{i}}}_{i = 1}^{n}$, 
  we get $\penc{\gtext}$ in $O(n^3)$ time.
  During the process we only maintain the latest $\Pref_{\ell}$ (discarding old ones), which takes $O(n^2)$ space to store.
  Given $\penc{\gtext}$, we can compute a p-string that p-matches $\gtext$ using Lemma~\ref{lemma:lex-smallest} in $O(n)$ time.
\end{proof}

\begin{figure}[t]
  \center{%
    \includegraphics[scale=0.38]{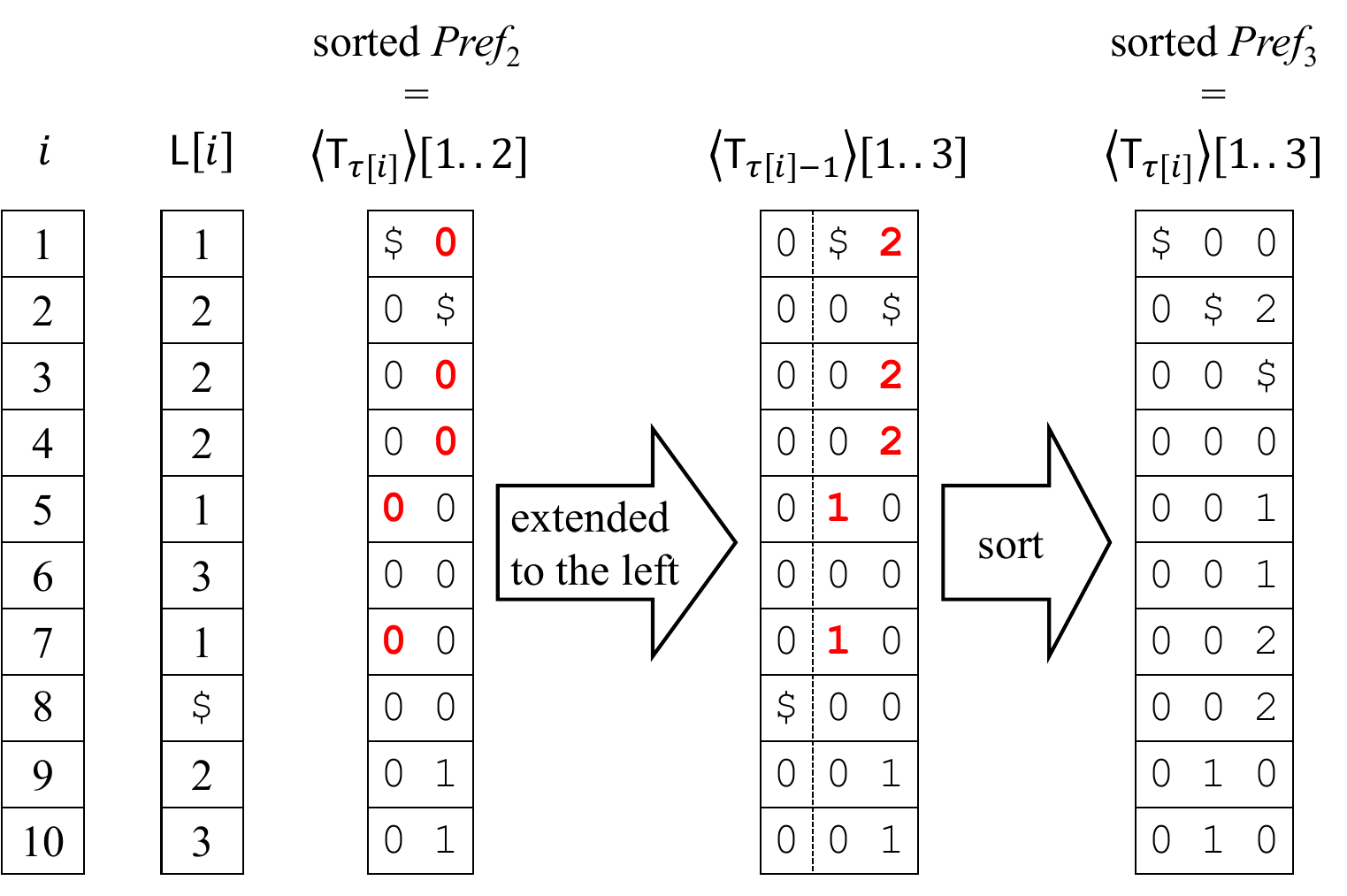}
  }
  \caption{Illustration for the computing process from $\penc{\gtext_{\psa[i]}}[1..2]$ to $\penc{\gtext_{\psa[i]}}[1..3]$ 
    for our running example $\gtext=\idtt{xyxzzxxyx\$}$ of Table~\ref{table:arrays}.
    When $\penc{\gtext_{\psa[i]}}[1..2]$ is extended to the left according to Proposition~\ref{prop:prepend}, 
    the 0s in red are turned into the distance to the beginning position in $\penc{\gtext_{\psa[i]-1}}[1..3]$.
  }
  \label{fig:n3}
\end{figure}

Next we show how to improve the algorithm to run in $O(n^2)$ time and $O(n)$ space.
From now on $\Pref_\ell$ is conceptual as we do not store them explicitly.
Instead we consider an integer array $\GroupA_{\ell}$ of length $n$ to record in $\GroupA_{\ell}[i]$ 
the tentative lexicographic rank of $\penc{\gtext_{\psa[i]-1}}$ at the moment we have processed $\Pref_\ell$
so that the lexicographic rank of $\penc{\gtext_{\psa[i]-1}}[1..\ell+1]$ can essentially be determined by 
radix sort with keys of the form $(\GroupA_{\ell}[i], \penc{\gtext_{\psa[i]-1}}[\ell+1])$.
Our algorithm gradually refines $\GroupA_{\ell}$ toward the LF-mapping while increasing $\ell$.
When $\GroupA_{\ell}$ becomes a permutation, a p-string that p-matches $\gtext$ is obtained using Lemma~\ref{lemma:lf2str}:
\begin{lemma}\label{lemma:lf2str}
  Given constant-time access to the pBWT $\Lstr$ and LF-mapping $\LF$ for a p-string $\gtext$ of length $n$, 
  we can compute in $O(n\frac{\lg\psigma}{\lg\lg\psigma})$ time a p-string $\gtext'$ that p-matches $\gtext$,
  where $\psigma \le n$ is the number of p-symbols used in $\gtext$.
\end{lemma}
\begin{proof}
  We determine $\gtext'[k]$ for $k = n, n-1, \dots, 1$ while keeping track of the position $i$ with $\psa[i] = k$ 
  using the formula $\psa[\LF[i]] = k-1 \mod n$.
  Recall that our starting point for $k = n$ is $i = 1$ and $\gtext'[n] = \$$.

  Suppose that we have determined $\gtext'[k..n]$, which contains $h_{k}$ distinct p-symbols from $\{ x_{h} \}_{h=1}^{h_{k}}$.
  At the moment we also have maintained a dynamic string for the permutation $Q$ of 
  $\{ x_{h} \}_{h=1}^{h_{k}}$ sorted in the increasing order of their leftmost occurrences in $\gtext'[k..n]$.
  Then we determine $\gtext'[k-1]$ and update $Q$ as follows:
  \begin{itemize}
    \item If $\Lstr[i] \in \sSigma$, we set $\gtext'[k-1] = \Lstr[i]$. We do nothing for $Q$.
    \item If $\Lstr[i] \in [1..\psigma]$ and $\Lstr[i] > h_{k}$, we use a new p-symbol $x_{h_{k}+1}$ for $\gtext'[k-1]$.
      We insert $x_{h_{k}+1}$ into the beginning of $Q$.
    \item If $\Lstr[i] \in [1..\psigma]$ and $\Lstr[i] \le h_{k}$, we set $\gtext'[k-1] = Q[\Lstr[i]]$.
      We delete $Q[\Lstr[i]]$ in $Q$ and reinsert it into the beginning of $Q$.
  \end{itemize}

  Since $\gtext'[k]$ is chosen according to Proposition~\ref{prop:prepend} so that $\penc{\gtext'[k..n]} = \penc{\gtext[k..n]}$,
  the resulting $\gtext'$ p-matches $\gtext$.
  Initially $Q$ is set to be $\emptystr$ and
  there are $O(n)$ insert/delete queries for $Q$ of length at most $\psigma$, 
  which takes a total of $O(n\frac{\lg\psigma}{\lg\lg\psigma})$ time using a dynamic string of Lemma~\ref{lem:drs}.
  Meanwhile other tasks can be done in $O(n)$ time.
\end{proof}

\begin{theorem}\label{theorem:n2}
  Given the pBWT $\Lstr$ of a p-string $\gtext$ of length $n$, 
  we can compute in $O(n^2)$ time and $O(n)$ space a p-string that p-matches $\penc{\gtext}$.
\end{theorem}
\begin{proof}
  For any $W \in \Pref_{\ell}$ let $\interval{W}$ be the maximal interval of p-encoded strings that are prefixed by $W$ in the sorted $\Pref_{\ell}$, i.e.,
  $\interval{W} = [|\multiset{ X \in \Pref_{\ell} \mid X < W }|+1..|\multiset{ X \in \Pref_{\ell} \mid X \le W }|]$.
  Since $\LF[i] \in \interval{\penc{\gtext_{\psa[i]-1}}[1..\ell]}$, we can determine $\LF[i]$ if the interval is a singleton.

  For some $\ell \in [1..n)$, suppose that $\GroupA_{\ell}[i]$ is set to the smallest integer of $\interval{\penc{\gtext_{\psa[i]-1}}[1..\ell]}$.
  Suppose also that we have two arrays $\ZeroA_{\ell}$ and $\EndA_{\ell}$ of length $n$ each such that 
  $\ZeroA_{\ell}[i] = \zeros{\penc{\gtext_{\psa[i]}}[1..\ell-1]}$ and $\EndA_{\ell}[i] = \penc{\gtext_{\psa[i]}}[\ell]$.
  See Fig.~\ref{fig:n2} for an example.
  Below we show how to compute $\GroupA_{\ell+1}$, $\ZeroA_{\ell+1}$ and $\EndA_{\ell+1}$ using $\GroupA_{\ell}$, $\ZeroA_{\ell}$ and $\EndA_{\ell}$ in $O(n)$ time.

  By Proposition~\ref{prop:prepend} we can compute $\penc{\gtext_{\psa[i]-1}}[\ell+1]$ in $O(1)$ time as
  \begin{equation*}
    \penc{\gtext_{\psa[i]-1}}[\ell+1] =
    \begin{cases}
      \ell            & \mbox{if $\ZeroA_{\ell}[i] = \Lstr[i]-1$ and $\EndA_{\ell}[i] = 0$}, \\
      \EndA_{\ell}[i]  & \mbox{otherwise.}
    \end{cases}
  \end{equation*}
  We then sort the set $\{ i \}_{i = 1}^{n}$ of indexes with keys of the form $(\GroupA_{\ell}[i], \penc{\gtext_{\psa[i]-1}}[\ell+1])$
  by radix sort in $O(n)$ time and space.
  While incrementing position $k$ on the sorted list $S$ of indexes, 
  we keep track of the smallest index $k'$ such that 
  $(\GroupA_{\ell}[S[k']], \penc{\gtext_{\psa[S[k']]-1}}[\ell+1]) = (\GroupA_{\ell}[S[k]], \penc{\gtext_{\psa[S[k]]-1}}[\ell+1])$,
  and set $\GroupA_{\ell+1}[S[k]] = k'$.
  Also, $\EndA_{\ell+1}[k] = \penc{\gtext_{\psa[S[k]]-1}}[\ell+1]$ and $\ZeroA_{\ell+1}[k] = \ZeroA_{\ell}[k] + b$,
  where $b = 1$ if $\penc{\gtext_{\psa[S[k]]-1}}[\ell+1] = 0$ and otherwise $b = 0$.
  Since all entries of $\GroupA_{\ell+1}$, $\ZeroA_{\ell+1}$ and $\EndA_{\ell+1}$ are filled in a single scan on $S$, it takes $O(n)$ time.

  The algorithm starts the above iterations from $\ell = 1$.
  First, we compute the sorted $\Pref_{1}$ in $O(n)$ time using the same procedure with the algorithm of Theorem~\ref{theorem:n3},
  which allows us to compute $\EndA_{1}[i]$ and $\GroupA_{1}[i]$ for every $i$.
  Obviously $\ZeroA_{1}$ is the zero initialized array.
  We then increment $\ell$ until $\GroupA_{\ell}$ becomes a permutation, which takes $O(n^2)$ time in total.
  Since the permutation is equivalent to the LF-mapping, we can get a p-string that p-matches $\gtext$ in $O(n^2)$ time using Lemma~\ref{lemma:lf2str}.
  During the process we only maintain the latest $\GroupA_{\ell+1}$, $\ZeroA_{\ell+1}$ and $\EndA_{\ell+1}$ (discarding old ones), which takes $O(n)$ space.
\end{proof}

\begin{figure}[t]
  \center{%
    \includegraphics[scale=0.38]{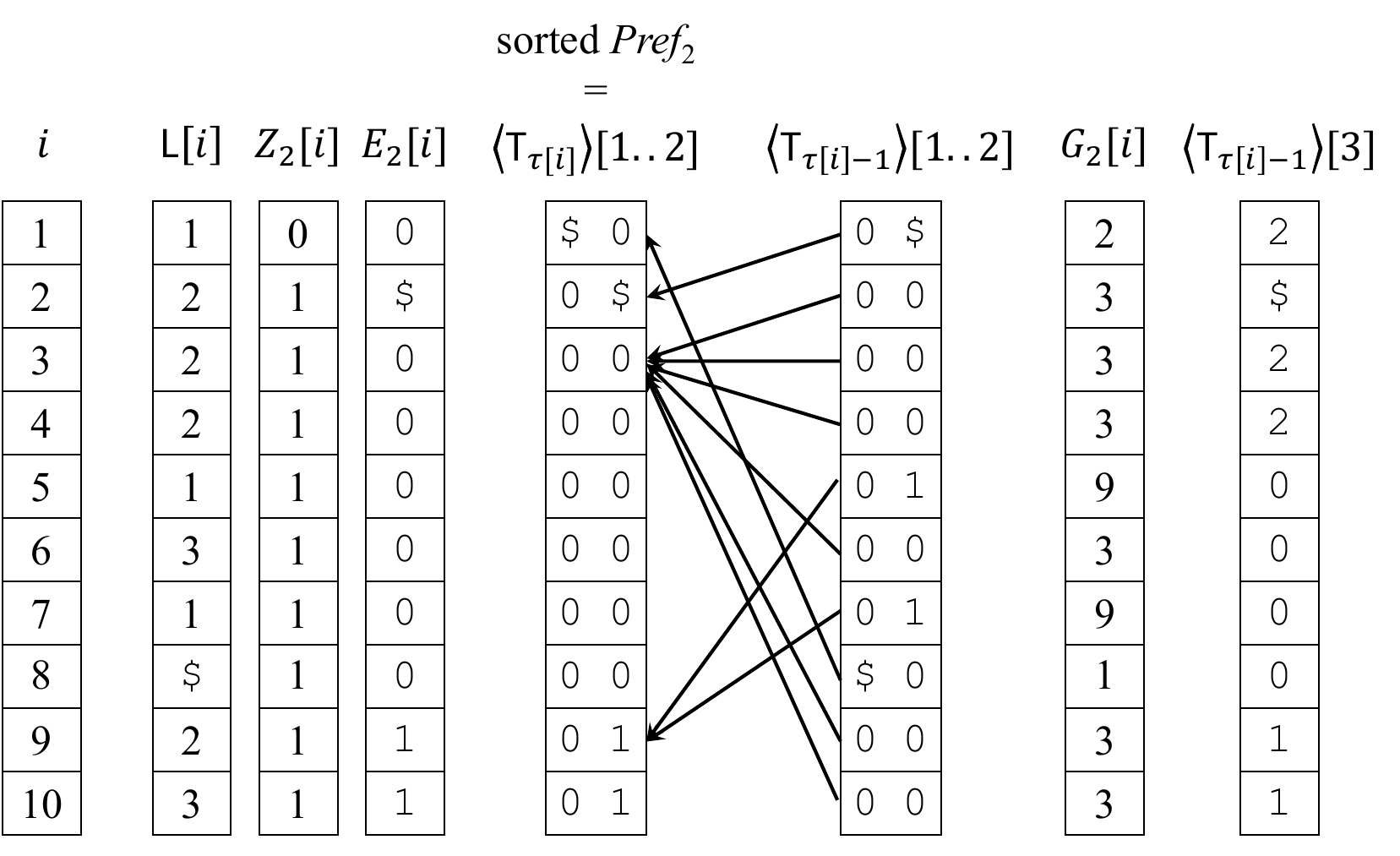}
  }
  \caption{Illustration of $\GroupA_{2}[i]$, $\ZeroA_{2}[i]$, $\EndA_{2}[i]$ and $\penc{\gtext_{\psa[i]-1}}[3]$
    for our running example $\gtext=\idtt{xyxzzxxyx\$}$ of Table~\ref{table:arrays}.
    $\GroupA_{2}[i]$ points to the smallest position at which $\penc{\gtext_{\psa[i]-1}}[1..2]$ appears in the sorted $\Pref_{2}$.
    Remark that we do not have $\penc{\gtext_{\psa[i]}}[1..2]$ and $\penc{\gtext_{\psa[i]-1}}[1..2]$ explicitly.
    Still we can compute $\penc{\gtext_{\psa[i]-1}}[3]$ from $\Lstr[i]$, $\ZeroA_{2}[i]$ and $\EndA_{2}[i]$.
    By sorting the pairs of the form $(\GroupA_{2}[i], \penc{\gtext_{\psa[i]-1}}[3])$, 
    we can obtain the refined intervals with prefixes of length 3.
  }
  \label{fig:n2}
\end{figure}

\section{Discussion}
In this paper, we showed that pBWTs are invertible in $O(n^2)$ time and $O(n)$ space.
The algorithm can easily be modified to work on a variant of pBWTs defined with prev$_{\infty}$-encoding of~\cite{2021KimC_SimplFmIndexForParam}
in which the leftmost occurrence of a p-symbol is replaced with the largest symbol $\infty$ instead of $0$.
It is known that the prev$_{\infty}$-encoding has a better property than the original prev-encoding to work on pBWTs:
On pBWTs with prev$_{\infty}$-encoding,
it holds that $\LF[i] < \LF[j]$ for any positions $i < j$ with $\Lstr[i] = \Lstr[j]$.
The modified version of our algorithm runs in $O(n^2)$ time and $O(n)$ space as it does not take advantage of the property of prev$_{\infty}$-encoding.

An obvious open question is whether $o(n^2)$-time inversion is possible.
A possible direction for improvement is to reduce the amount of steps needed to infer the LF-mapping.
For example, there are some cases where we are able to pinpoint $\LF[i]$ 
even if $\interval{\penc{\gtext_{\psa[i]-1}}[1..\ell]}$ contains two or more integers, using the following lemma.
\footnote{The property has been used in pBWT-based data structures (e.g., see cases (b) and (c1) of Lemma 1 of~\cite{2017GangulyST_PbwtAchievSuccinDataStruc_SODA}). We rephrase it to fit our context and give a proof for completeness.}
\begin{lemma}\label{lemma:lf-fixed}
  For any $\ell \in [1..n]$ and $i \in [1..n)$, 
  let $W = \penc{\gtext_{\psa[i]}}[1..\ell-1]$ and $\hat{W} = \penc{\gtext_{\psa[i]-1}}[1..\ell]$.
  If $\Lstr[i] \in \sSigma$ or $\Lstr[i] \le \zeros{W}$, then it holds that
  \[
  \LF[i] = |\{ X \in \Pref_{\ell+1} \mid X < \hat{W} \}| + |\{ j \in [1..i] \mid \penc{\gtext_{\psa[j]-1}}[1..\ell] = \hat{W} \}|.
  \]
\end{lemma}
\begin{proof}
  It is enough to show that $\penc{\gtext_{\psa[j]-1}} < \penc{\gtext_{\psa[j']-1}}$ for any positions $j < j'$ with $\penc{\gtext_{\psa[j]-1}}[1..\ell] = \penc{\gtext_{\psa[j']-1}}[1..\ell] = \hat{W}$.
  Let $k = \lcp(\gtext_{\psa[j]}, \gtext_{\psa[j']})$.
  Note that $k$ is at least $\ell - 1$ since $\penc{\gtext_{\psa[j]}}[1..\ell-1] = \penc{\gtext_{\psa[j']}}[1..\ell-1] = W$ 
  is necessary to have $\penc{\gtext_{\psa[j]-1}}[1..\ell] = \penc{\gtext_{\psa[j']-1}}[1..\ell] = \hat{W}$.
  Also $k$ is at most $n-1$ due to $\gtext_{\psa[j]} \neq \gtext_{\psa[j']}$.

  If $\Lstr[i] \in \sSigma$, it follows from Proposition~\ref{prop:prepend} that
  $\penc{\gtext_{\psa[j]-1}} = \Lstr[j]\penc{\gtext_{\psa[j]}}[1..n-1]$ 
  and $\penc{\gtext_{\psa[j]-1}} = \Lstr[j']\penc{\gtext_{\psa[j']}}[1..n-1]$.
  If $\Lstr[i] \le \zeros{W}$, it follows from Proposition~\ref{prop:prepend} that
  $\penc{\gtext_{\psa[j]-1}} = \hat{W} \cdot \penc{\gtext_{\psa[j]}}[\ell..n-1]$ 
  and $\penc{\gtext_{\psa[j']-1}} = \hat{W} \cdot \penc{\gtext_{\psa[j']}}[\ell..n-1]$
  with $\hat{W}$ being obtained by modifying the $\Lstr[i]$-th $0$ in $W$ and prepend $0$.
  In either case, we have
  $\lcp(\penc{\gtext_{\psa[j]-1}}, \penc{\gtext_{\psa[j']-1}}) = k+1$ and 
  $\penc{\gtext_{\psa[j]-1}}[k+1] = \penc{\gtext_{\psa[j]}}[k] < \penc{\gtext_{\psa[j']}}[k] = \penc{\gtext_{\psa[j']-1}}[k]$,
  which implies that $\penc{\gtext_{\psa[j]-1}} < \penc{\gtext_{\psa[j']-1}}$.
\end{proof}

Using properties of this kind, one might hope to improve to $O(n \psigma)$ time,
which is often seen in time complexities for p-matching algorithms 
like direct construction of pSAs~\cite{2019FujisatoNIBT_DirecLinearTimeConstOf_SPIRE}.

\section*{Acknowledgements}
Tomohiro I was supported by JSPS KAKENHI (Grant Numbers 22K11907 and 24K02899) and JST AIP Acceleration Research JPMJCR24U4, Japan.

\bibliographystyle{plainurl}
\bibliography{refs}

\begin{thebibliography}{10}

\bibitem{2024AmirKLMS_ReconParamStrinFromParam}
Amihood Amir, Eitan Kondratovsky, Gad~M. Landau, Shoshana Marcus, and Dina
  Sokol.
\newblock Reconstructing parameterized strings from parameterized suffix and
  {LCP} arrays.
\newblock {\em Theor. Comput. Sci.}, 981:114230, 2024.

\bibitem{2024AmirKMS_LinearTimeReconOfParam_SPIRE}
Amihood Amir, Eitan Kondratovsky, Shoshana Marcus, and Dina Sokol.
\newblock Linear time reconstruction of parameterized strings from
  parameterized suffix and {LCP} arrays for constant-sized alphabets.
\newblock In {\em Proc. 31st International Symposium on String Processing and
  Information Retrieval ({SPIRE}) 2024}, volume 14899 of {\em Lecture Notes in
  Computer Science}, pages 1--15. Springer, 2024.

\bibitem{1993Baker_TheorOfParamPatterMatch_STOC}
Brenda~S. Baker.
\newblock A theory of parameterized pattern matching: algorithms and
  applications.
\newblock In S.~Rao Kosaraju, David~S. Johnson, and Alok Aggarwal, editors,
  {\em Proc. 25th Annual {ACM} Symposium on Theory of Computing ({STOC})},
  pages 71--80. {ACM}, 1993.
\newblock \href {https://doi.org/10.1145/167088.167115}
  {\path{doi:10.1145/167088.167115}}.

\bibitem{1996Baker_ParamPatterMatchAlgorAnd}
Brenda~S. Baker.
\newblock Parameterized pattern matching: Algorithms and applications.
\newblock {\em Journal of Computer and System Sciences}, 52(1):28--42, 1996.
\newblock \href {https://doi.org/10.1006/jcss.1996.0003}
  {\path{doi:10.1006/jcss.1996.0003}}.

\bibitem{1997Baker_ParamDuplicInStrinAlgor}
Brenda~S. Baker.
\newblock Parameterized duplication in strings: Algorithms and an application
  to software maintenance.
\newblock {\em {SIAM} J. Comput.}, 26(5):1343--1362, 1997.
\newblock \href {https://doi.org/10.1137/S0097539793246707}
  {\path{doi:10.1137/S0097539793246707}}.

\bibitem{2003BannaiIST_InferStrinFromGraphAnd_MFCS}
Hideo Bannai, Shunsuke Inenaga, Ayumi Shinohara, and Masayuki Takeda.
\newblock Inferring strings from graphs and arrays.
\newblock In {\em Proc. 28th International Symposium on Mathematical
  Foundations of Computer Science ({MFCS}) 2003}, pages 208--217, 2003.

\bibitem{Burrows1994BWT}
Michael Burrows and David~J Wheeler.
\newblock A block-sorting lossless data compression algorithm.
\newblock Technical report, HP Labs, 1994.

\bibitem{2014CazauxR_ReverEnginOfCompacSuffix}
Bastien Cazaux and Eric Rivals.
\newblock Reverse engineering of compact suffix trees and links: {A} novel
  algorithm.
\newblock {\em J. Discrete Algorithms}, 28:9--22, 2014.

\bibitem{2008DeguchiHBIT_ParamSuffixArrayForBinar_PSC}
Satoshi Deguchi, Fumihito Higashijima, Hideo Bannai, Shunsuke Inenaga, and
  Masayuki Takeda.
\newblock Parameterized suffix arrays for binary strings.
\newblock In {\em Proc. Prague Stringology Conference ({PSC}) 2008}, pages
  84--94, 2008.
\newblock URL: \url{http://www.stringology.org/event/2008/p08.html}.

\bibitem{Ferragina2000ODS}
Paolo Ferragina and Giovanni Manzini.
\newblock Opportunistic data structures with applications.
\newblock In {\em Proc. 41st Annual Symposium on Foundations of Computer
  Science ({FOCS}) 2000}, pages 390--398, 2000.
\newblock \href {https://doi.org/10.1109/SFCS.2000.892127}
  {\path{doi:10.1109/SFCS.2000.892127}}.

\bibitem{2019FujisatoNIBT_DirecLinearTimeConstOf_SPIRE}
Noriki Fujisato, Yuto Nakashima, Shunsuke Inenaga, Hideo Bannai, and Masayuki
  Takeda.
\newblock Direct linear time construction of parameterized suffix and {LCP}
  arrays for constant alphabets.
\newblock In Nieves~R. Brisaboa and Simon~J. Puglisi, editors, {\em Proc. 26th
  International Symposium on String Processing and Information Retrieval
  ({SPIRE}) 2019}, volume 11811 of {\em Lecture Notes in Computer Science},
  pages 382--391. Springer, 2019.
\newblock \href {https://doi.org/10.1007/978-3-030-32686-9_27}
  {\path{doi:10.1007/978-3-030-32686-9_27}}.

\bibitem{2018GagieNP_OptimTimeTextIndexIn_SODA}
Travis Gagie, Gonzalo Navarro, and Nicola Prezza.
\newblock Optimal-time text indexing in bwt-runs bounded space.
\newblock In {\em In Proc. 29th Annual {ACM-SIAM} Symposium on Discrete
  Algorithms ({SODA}) 2018}, pages 1459--1477, 2018.

\bibitem{2017GangulyST_PbwtAchievSuccinDataStruc_SODA}
Arnab Ganguly, Rahul Shah, and Sharma~V. Thankachan.
\newblock {pBWT}: Achieving succinct data structures for parameterized pattern
  matching and related problems.
\newblock In {\em Proc. 28th Annual {ACM-SIAM} Symposium on Discrete Algorithms
  ({SODA}) 2017}, pages 397--407, 2017.
\newblock \href {https://doi.org/10.1137/1.9781611974782.25}
  {\path{doi:10.1137/1.9781611974782.25}}.

\bibitem{2022GangulyST_FullyFunctParamSuffixTrees_ICALP}
Arnab Ganguly, Rahul Shah, and Sharma~V. Thankachan.
\newblock Fully functional parameterized suffix trees in compact space.
\newblock In Mikolaj Bojanczyk, Emanuela Merelli, and David~P. Woodruff,
  editors, {\em Proc. 49th International Colloquium on Automata, Languages, and
  Programming, ({ICALP}) 2022}, volume 229 of {\em LIPIcs}, pages 65:1--65:18.
  Schloss Dagstuhl - Leibniz-Zentrum f{\"{u}}r Informatik, 2022.
\newblock \href {https://doi.org/10.4230/LIPIcs.ICALP.2022.65}
  {\path{doi:10.4230/LIPIcs.ICALP.2022.65}}.

\bibitem{2011IIBT_VerifAndEnumerParamBorder}
Tomohiro I, Shunsuke Inenaga, Hideo Bannai, and Masayuki Takeda.
\newblock Verifying and enumerating parameterized border arrays.
\newblock {\em Theoretical Computer Science}, 412(50):6959 -- 6981, 2011.
\newblock URL:
  \url{http://www.sciencedirect.com/science/article/pii/S0304397511007742},
  \href {https://doi.org/10.1016/j.tcs.2011.09.008}
  {\path{doi:10.1016/j.tcs.2011.09.008}}.

\bibitem{2014IIBT_InferStrinFromSuffixTrees}
Tomohiro I, Shunsuke Inenaga, Hideo Bannai, and Masayuki Takeda.
\newblock Inferring strings from suffix trees and links on a binary alphabet.
\newblock {\em Discret. Appl. Math.}, 163:316--325, 2014.

\bibitem{2024IseriIHKYS_BreakBarrierInConstCompac_ICALP}
Kento Iseri, Tomohiro I, Diptarama Hendrian, Dominik K{\"{o}}ppl, Ryo
  Yoshinaka, and Ayumi Shinohara.
\newblock Breaking a barrier in constructing compact indexes for parameterized
  pattern matching.
\newblock In {\em Proc. 51st International Colloquium on Automata, Languages,
  and Programming ({ICALP}) 2024}, volume 297 of {\em LIPIcs}, pages
  89:1--89:19. Schloss Dagstuhl - Leibniz-Zentrum f{\"{u}}r Informatik, 2024.

\bibitem{2023KaerkkaeinenPP_StrinInferFromLongesCommon}
Juha K{\"{a}}rkk{\"{a}}inen, Marcin Piatkowski, and Simon~J. Puglisi.
\newblock String inference from longest-common-prefix array.
\newblock {\em Theor. Comput. Sci.}, 942:180--199, 2023.

\bibitem{2021KimC_SimplFmIndexForParam}
Sung{-}Hwan Kim and Hwan{-}Gue Cho.
\newblock Simpler {FM}-index for parameterized string matching.
\newblock {\em Inf. Process. Lett.}, 165:106026, 2021.
\newblock \href {https://doi.org/10.1016/j.ipl.2020.106026}
  {\path{doi:10.1016/j.ipl.2020.106026}}.

\bibitem{2013KucherovTV_CombinOfSuffixArray}
Gregory Kucherov, Lilla T{\'{o}}thm{\'{e}}r{\'{e}}sz, and St{\'{e}}phane
  Vialette.
\newblock On the combinatorics of suffix arrays.
\newblock {\em Inf. Process. Lett.}, 113(22-24):915--920, 2013.

\bibitem{2020MendivelsoTP_BriefHistorOfParamMatch}
Juan Mendivelso, Sharma~V. Thankachan, and Yoan~J. Pinz{\'{o}}n.
\newblock A brief history of parameterized matching problems.
\newblock {\em Discret. Appl. Math.}, 274:103--115, 2020.
\newblock \href {https://doi.org/10.1016/j.dam.2018.07.017}
  {\path{doi:10.1016/j.dam.2018.07.017}}.

\bibitem{2015MunroN_ComprDataStrucForDynam_ESA}
J.~Ian Munro and Yakov Nekrich.
\newblock Compressed data structures for dynamic sequences.
\newblock In {\em Proc. 23rd Annual European Symposium on Algorithms ({ESA})
  2015}, pages 891--902, 2015.
\newblock \href {https://doi.org/10.1007/978-3-662-48350-3_74}
  {\path{doi:10.1007/978-3-662-48350-3_74}}.

\bibitem{Starikovskaya2015sto}
Tatiana~A. Starikovskaya and Hjalte~Wedel Vildh{\o}j.
\newblock A suffix tree or not a suffix tree?
\newblock {\em J. Discrete Algorithms}, 32:14--23, 2015.

\end{thebibliography}

\end{document}